\documentclass{article}
\usepackage[utf8]{inputenc}
\usepackage{comment}
\usepackage{amsfonts, makecell}
\usepackage{amsmath}
\usepackage{amssymb}
\usepackage{amsthm}
\usepackage{mathtools}
\usepackage{caption}
\RequirePackage{doi}
\usepackage{hyperref}
\usepackage{color}
\usepackage{arydshln}
\usepackage[lined,ruled,linesnumbered,noend]{algorithm2e}
\usepackage{algpseudocode}
\usepackage{graphicx}
\usepackage{float}
\usepackage[scr]{rsfso}
\usepackage[a4paper, total={6in, 8in}]{geometry}
\usepackage{cleveref}
\newtheorem{theorem}{Theorem}
\newtheorem{claim}[theorem]{Claim}

\newtheorem{corollary}[theorem]{Corollary}
\newtheorem{definition}[theorem]{Definition}
\newtheorem{example}[theorem]{Example}
\newtheorem{lemma}[theorem]{Lemma}
\newtheorem{proposition}[theorem]{Proposition}
\newtheorem{remark}[theorem]{Remark}
\newtheorem{observation}[theorem]{Observation}
\newtheorem{note}[theorem]{Note}

\newcommand{\FF}{\mathbb{F}}
\newcommand{\NN}{\mathbb{N}}

\newcommand{\C}{\mathcal{C}}
\renewcommand{\P}{\mathcal{P}}

\newcommand{\chr}{\operatorname{char}}
\newcommand{\sing}{\operatorname{sing}}

\newcommand{\esym}[2]{e_{#2}^{#1}}

\usepackage[textwidth=2.5cm]{todonotes}
\newcommand{\mnote}[3]{\todo[color=#3!40,size=\footnotesize]{\textbf{#2:} #1}}
\newcommand{\inote}[3]{\todo[color=#3!40,inline]{\textbf{#2:} #1}}
\newcommand{\ian}[1]{\mnote{#1}{IO}{red}}
\newcommand{\iian}[1]{\inote{#1}{IO}{red}}

\def\multiset#1#2{\ensuremath{\left(\kern-.3em\left(\genfrac{}{}{0pt}{}{#1}{#2}\right)\kern-.3em\right)}}

\numberwithin{theorem}{section}

\title{Computing the Elementary Symmetric Polynomials in Positive Characteristics}
\author{Ian Orzel\footnote{Email: iano@di.ku.dk. This work was funded by the European Research Council (ERC) under grant agreement no. 101125652 (ALBA).}\\Department of Computer Science\\University of Copenhagen\\Copenhagen, Denmark}

\begin{document}

\maketitle

\begin{abstract}
We first extend the results of Chatterjee, Kumar, Shi, Volk (Computational Complexity 2022)
by showing that the degree-$d$ elementary symmetric polynomials in $n$ variables have arithmetic formula lower bounds of $\Omega(d(n-d))$ over fields of positive characteristic.
This is done by characterizing the dimension of the singular locus of the elementary symmetric polynomials in positive characteristic fields (parts of the characterization only hold over infinite fields).
Then, we show that the results of the universality of border fan-in-two $\Sigma\Pi\Sigma$-circuits from Kumar (ACM Trans. Comput. Theory 2020) over zero-characteristic fields do not extend to fields of positive characteristic.
This is done through studying linear projections of the elementary symmetric polynomial as studied by Shpilka (Journal of Computer and System Sciences 2002).
\end{abstract}

\newpage
\section{Introduction}
Given $n$ independent variables $x_1, \ldots, x_n$ and a degree $d \le n$, we define the \emph{elementary symmetric polynomial} of degree-$d$ (over an arbitrary field $\FF$ of characteristic denoted by $\operatorname{char}(\FF)$) to be the sum of all degree $d$ multilinear monomials, i.e.,
\[
\esym{n}{d} = \sum_{S \in {[n] \choose d}}\prod_{i \in S}x_i.
\]
This polynomial naturally appears in many situations; for instance, consider the product of affine polynomials.
\begin{equation}\label{eq:prod_to_esp}
\prod_{i=1}^n(x + a_i) = \sum_{k = 0}^nx^{n-k}\esym{n}{k}(a_1, \ldots, a_n).
\end{equation}
We include more fundamental properties in \Cref{sec:esp}.

Specifically, in this article, we will focus on these polynomials from the lens of computation.
We will define computation from the arithmetic lens; recall that an arithmetic circuit is a directed acyclic graph defining an arithmetic computation, where vertices represent either inputs, given as elements of $\FF$ or $\{x_1, \ldots, x_n\}$, addition, or multiplication.
We then say that an arithmetic formula is an arithmetic circuit whose underlying graph is a tree.
Our first result is to extend the results of \cite{CKSV22} by showing they apply to fields of positive characteristic.
\begin{theorem}\label{thm:alg_formulas}
For an arbitrary field, any algebraic formula computing $\esym{n}{d}$ has size $\Omega(d(n-d))$.
\end{theorem}

For the second result, we recall the computational model defined by $\Sigma^{[k]}\Pi\Sigma$, which we can view as a computation represented by the sum of $k$ products of affine polynomials, i.e., polynomials written as $\sum_{i=1}^k\prod_{j}\ell_{i,j}$, for affine $\ell_{i,j} \in \FF[x_1, \ldots, x_n]$.
We can further consider the class $\overline{\Sigma^{[k]}\Pi\Sigma}$, which is the same as the previous class, but we now allow for approximation for some natural definition of approximation (which is known as border complexity, see \Cref{subsec:border_comp}).
Then, it was proven in \cite{Kum20} that, over fields of characteristic zero, all polynomials can be computed by the class $\overline{\Sigma^{[2]}\Pi\Sigma}$.
The second main result of this paper is that this phenomenon does not extend to fields of positive characteristic.
\begin{theorem}\label{thm:main}
If $\operatorname{char}(\FF) \ne 0$, then there exists a polynomial in at most $n$ variables that can only be computed by $\overline{\Sigma^{[k]}\Pi\Sigma}$ for $k = \Omega(n)$.
\end{theorem}

\subsection{Motivations}
\subparagraph*{Algebraic Complexity.} The complexity of the elementary symmetric polynomials is a well-studied problem in the field of algebraic complexity theory.
One of the best-known algorithms for computing the elementary symmetric polynomial is a $\Sigma\Pi\Sigma$ formula of size $O(n^2)$ (see \cite{Shp02}), which is often known as the Ben-Or construction.
Over fields of characteristic zero, a lower bound of $\Omega(d(n-d))$ was shown in \cite{Shp02} for $\Sigma\Pi\Sigma$ formulas and was then later expanded to general algebraic formulas in \cite{CKSV22}.
This lower bound is tight for certain selections of $n$ and $d$.

The complexity of these polynomials have naturally appeared recently in \cite{LST22} and \cite{FLST24} in the study of constant-depth formulas.
The problem here is the homogenization problem, where one takes a $\Sigma\Pi\Sigma$ formula and tries to convert it into a homogeneous formula (a formula where every node computes a homogeneous polynomial) whose depth is still constant.
The strategy is then to apply \Cref{eq:prod_to_esp} to each product gate and then replace computations of the elementary symmetric polynomials with homogeneous computations of them (they use the construction of \cite{SW01} for $\Sigma\Pi\Sigma\Pi$ formulas of size $n2^{O(\sqrt{d})}$, which requires the field to be characteristic zero).
This idea can be further expanded to larger depths.

\subparagraph*{Border Complexity.} Computing the elementary symmetric polynomials also naturally appears in the problem of border complexity, where we, rather than requiring a polynomial to be computed exactly, only require it to be approximated for some natural definition of approximation.
Interest in border complexity has increased due to the introduction of the geometric complexity theory program, which hopes to use techniques from algebraic geometry and representation theory to give border-complexity lower bounds.
One reason why this notion is considered natural is that many lower-bound techniques also give lower bounds in the border setting, meaning that current techniques have trouble differentiating these models.
In spite of this, border complexity adds a lot of power; one such example of this can be seen in \cite{Kum20}, where it was proved that every polynomial can be approximated by a computation of the form $\Pi_iL_i - 1$ in fields of characteristic zero, where each $L_i$ is an affine polynomial.
This fact is not true in the non-border setting, as some polynomials can only be computed by $\Sigma^{[k]}\Pi\Sigma$ formulas where $k = \Omega(n)$.
Interestingly, this result is intimately-related to the computation of elementary symmetric polynomials, where we can view the elementary symmetric polynomials as a computational model.
The key element from the proof is the observation that, for any homogeneous degree-$d$ polynomial $f$, $f$ can be written as a linear projection of an elementary symmetric polynomial in fields of characteristic zero, i.e., there are linear $L_1, \ldots, L_m \in \FF[x_1, \ldots, x_n]$ such that $f = e_d^m(L_1, \ldots, L_m)$.

\subparagraph*{Positive Characteristic.} In this paper, we will study these problems over fields of positive characteristic.
Many results focus on fields of characterisic zero, as they are generally seen as more ``natural'' fields (recall that the real and complex numbers have characteristic zero).
Specifically, many proofs require division by natural numbers, causing them to fail in positive characteristics.
For this reason, many results in algebraic complexity theory only work in fields of characteristic zero.
This paper is focused on extending some of these results to fields of arbitrary characteristic, but it is also natural to ask why we should study this case.

One simple reasoning follows from the underlying algebraic-complexity-theory program, whose major goal is to show that VP is not equal to VNP (an algebraic analogue to proving that P is not equal to NP).
From \cite{Bur00}, we see that it is sufficient to show that VP equals VNP in a field of any characteristic to collapse the polynomial hierarchy.
In fact, classical complexity problems are usually posed over alphabets described by $\{0, 1\}$, which can be viewed as a field of characteristic two.
It is thus not only interesting to study positive characteristics, but, fundamentally, our strategy to approach this problem \emph{should} be independent of the underlying field.

Interestingly, the issue of positive characteristic also shows up in the field of proof complexity, whose focus is on developing lower bounds for the size of ``proofs'' of logical statements in various proof systems.
Specifically, one major proof system studied is known as the $\text{AC}^0[p]$-Frege proof system, which is a bounded-depth proof system with $\operatorname{MOD}_p$ gates, where a major open problem is to find superpolynomial lower bounds.
The theory of proof complexity was brought into the world of algebraic complexity through the introduction of the Ideal Proof System (IPS), which defines its proofs using algebraic circuits.
In the influential work \cite{GP18}, it was shown that superpolynomial lower bounds for IPS over fields of characteristic $p \ne 0$ give superpolynomial $\text{AC}^0[p]$-Frege lower bounds.
They further showed that superpolynomial circuit lower bounds lead to superpolynomial IPS lower bounds; hence, finding such lower bounds over fields of nonzero characteristic can lead to solving this open problem in proof complexity.

Another field containing interesting applications comes from polynomial identity testing (PIT).
Given a polynomial belonging to some class of polynomials in a predefined representation, we ask whether there is a \emph{deterministic} algorithm to determine if it is identically the zero polynomial.
One major line of research in algebraic complexity theory is finding polynomial-time PIT algorithms for interesting subclasses of polynomials.
As it turns out, the problem of PIT and lower bounds are intimately related; finding lower bounds often leads to PIT algorithms.
For instance, in \cite{KI04}, it was shown that exponential lower bounds for circuits computing the permanent would lead to quasi-polynomial-time PIT algorithms for polynomial-sized circuits.
Unfortunately, if the lower bounds discovered only apply for fields of characteristic zero, then so will the PIT algorithms.
If we would like our PIT algorithms to work over any field, we also need our lower-bound strategies to work in these cases.

Recently, there has been a theme of extending lower-bound results from fields of characteristic zero to fields of positive characteristic.
This trend can be seen in papers such as \cite{And20}, \cite{DIKMNT24}, \cite{For24}, \cite{BLRS25}, and \cite{BKRRSS25}.

\subsection{Proof techniques}
\subparagraph*{Proof of Theorem 1.} We will first focus on proving Theorem \ref{thm:alg_formulas} in Section \ref{sec:order_two_zeros}.
We begin by analyzing the results of \cite{CKSV22} to isolate the part of the proof that requires the field characteristic to be zero.
It turns out that this comes from the analysis of what we will call the set of order-two zeros of a polynomial.

Algebraic complexity theory focuses on using the algebraic properties of polynomials to split them into classes based on how ``hard" they are to compute.
One such property that has been utilized (see \cite{vzG87}, \cite{Kum19}, \cite{CKSV22}, \cite{KV22}, \cite{ABV17}) is what we will call the \emph{order-$2$-zero space}\footnote{This is often referred to as the singular locus. See Remark \ref{rem:sing_locus}}.
Given a polynomial $f \in \FF[x_1, \ldots, x_n]$, we define its order-two-zero space, denoted $V_2(f)$, to be the points where it vanishes along with its first-order partial derivatives.

The main use of this object in recent results is for relating it to the size of certain types of decompositions.
We summarize these results in the following lemma.
\begin{lemma}[See Lemma 1.7 of \cite{Kum19} and Lemma 3.4 of \cite{CKSV22}]\label{lem:kumar_lemma}
Suppose $\FF$ is algebraically closed.
Let $f \in \FF[x_1, \ldots, x_n]$ be homogeneous of degree $d$. If there are constant-free polynomials $f_1, g_1, \ldots, f_k, g_k \in \FF[x_1, \ldots, x_n]$ (meaning that their constant parts are zero) and a polynomial $h \in \FF[x_1, \ldots, x_n]$ with $\deg(h) < d$ such that
\[
f = \sum_{i=1}^kf_ig_i + h,
\]
then $\dim V_2(f) \ge n - 2k$.
\end{lemma}
\begin{note}
Notice that the assumption that the field is algebraically closed is necessary for the previous lemma.
For example, consider $f = (x_1^2 + \dots + x_n^2)^2$.
We notice that $V_2(f) = V(f) = \{0\}$ over the real numbers, but we get a lower bound of $\dim V_2(f) \ge n - 2$ over the complex numbers.
\end{note}

The key observation for the proof in \cite{CKSV22} was that small formulas computed polynomials with a large order-two-zero space.
Then, they prove that the dimension of the order-two-zero space of $e_d^n$ is exactly $(d-2)$ (note that only the upper bound is necessary for their proof).
The key step in this proof is to use the fact that, for a homogeneous degree-$d$ polynomial $f$, $\sum_{i=1}^nx_i \cdot \frac{\partial f}{\partial x_i} = d \cdot f$.
Unfortunately, if the characteristic of the underlying field is not zero, this might not provide helpful information.

The main contribution of this paper is to characterize the dimension of the order-two-zero space over fields of positive characteristic.
In this situation, the upper bound of $(d-2)$ turns out to no longer hold; but the statement is true if we relax our attention to $(d-1)$ (which is good enough for the lower bounds).
\begin{lemma}\label{lem:order_two_esp}
For $p \coloneq \operatorname{char}(\FF) \ne 0$, $\FF$ an infinite field, and $2 \le d \le n$, we have that $d-2 \le \dim V_2(\esym{n}{d}) \le d-1$.
Specifically, if $n - d + 1 \equiv 0 \mod p^{\lfloor \log_p d \rfloor+1}$, then $\dim V_2(\esym{n}{d}) = d-1$.
Otherwise, $\dim V_2(\esym{n}{d}) = d-2$.
\end{lemma}
We note here that the authors have been made aware that this result was proved simultaneously in a recent preprint, see \cite{Xie26}.

\subparagraph*{Proof of Theorem 2.} We will now turn our attention to the proof of Theorem \ref{thm:main}, which we prove in Section \ref{sec:symmetric}.
To prove this, we will use an object that was studied in \cite{Shp02}: linear projections of elementary symmetric polynomials.
It was shown that, for every homogeneous degree-$d$ polynomial, there are linear forms $L_1, \ldots, L_m \in \FF[x_1, \ldots, x_n]$ such that $\esym{m}{d}(L_1, \ldots, L_m) = f$ (they, in fact, prove the stronger statement that, for each $k < d$, it is also true that $\esym{m}{k}(L_1, \ldots, L_m) = 0$).
This fact is fundamental in proving the result in \cite{Kum20}.
The first main step in our prove is show that this connection can also be used for negative results.

Specifically, we prove that every polynomial that can be computed by $\overline{\Sigma^{[k]}\Pi\Sigma}$ can also be approximated by the sum of $k$ linear projections of the elementary symmetric polynomials.
It then suffices to study these linear projections.
In fact, when we look at just the model from \cite{Kum20}, it suffices to only consider a single linear projection.

Consider $\operatorname{char}(\FF) = 2$.
The idea to prove the lower bound is to apply the Newton Identities (\ref{eq:newton_identity}) of degree three, which state that
\[
\esym{m}{3}(L_1, \ldots, L_m) = (L_1 + \dots + L_m)e_2^m(L_1, \ldots, L_m) + (\esym{m}{1}(L_1, \ldots, L_m))^3 + L_1^3 + \ldots + L_m^3.
\]
In other words, assuming we select some $f \in \FF[x_1, \ldots, x_n]$ that is a degree-three form, we can write
\[
f = L \cdot Q + f_1^3 + \ldots + f_m^3,
\]
where $L$ is a linear form, $Q$ is a quadratic form, and $f_1, \ldots, f_m$ are linear forms.
The key here will be to use the ``Freshman Dream'' inequality, which states that $(x + y)^2 = x^2 + y^2$, from which we conclude that no monomial in an $f_i^3$ is multilinear.
Thus, if we choose $f$ to be a multilinear polynomial, it is ``essentially'' irreducible.
So we choose $f$ to be something like $x_1y_1z_1 + x_2y_2z_2$, and we complete the proof by showing that $f$ is not ``essentially'' irreducible (a contradiction).

We have, of course, missed the most important part of the argument: this also needs to be true in the border setting.
At first, it may seem like nothing has been gained; our original problem was trying to show that you cannot approximate some polynomial by some product of linear forms.
But now we have the luxury that our identity is homogeneous; whereas, the product of linear forms could have an exponential number of such forms in the product (which it does in the original construction in \cite{Kum20}).
We now have that $f$ is ``essentially'' approximated by the product of a linear thing and a quadratic thing, which is very easy to analyze.
To do this, we provide a very simple and easy-to-state meta-polynomial (meaning a polynomial in the coefficients of another polynomial) that we use to derive a contradiction.

\begin{note}
In the proofs of some of the theorems, we will sometimes require our field to be algebraically closed.
Because the main results from this paper are lower bounds, these results also apply to fields that are not algebraically-closed fields, as every field can be extended into one that is algebraically closed.
\end{note}

\subsection{Summary}
In \Cref{sec:prelim}, we begin by providing some basic definitions and notations that will be used for the rest of the paper.
In \Cref{sec:order_two_zeros}, we will focus on proving \Cref{thm:alg_formulas}; specifically, we will present an upper bound for the order-two-zero space of $e_d^n$; we show that $\dim V_2(e_d^n) \le d-1$.
In \Cref{sec:symmetric}, we focus on proving \Cref{thm:main}, but we will only prove a subcase where $\chr(\FF) = 2$ and the top fan-in of our circuits is one; we prove the general case in the appendix.

We then provide a short summary of the sections of the appendix.
In \Cref{sec:zero_space_anal}, we will prove \Cref{lem:order_two_esp} in its entirety.
Then, in \Cref{sec:comp_sym_model}, we provide a few proofs that were skipped in \Cref{sec:order_two_zeros}.
In \Cref{sec:sym_model_omit}, we provide a full proof of \Cref{thm:main} in all its generality.
In \Cref{sec:reduction}, we provide an explanation of why it is sufficient to consider symmetric computation instead of $\Sigma\Pi\Sigma$ formulas to prove \Cref{thm:main}.
Finally, in \Cref{sec:formula_lower_bounds}, we reprove the results of \cite{CKSV22} to clarify why it only suffices to study the order-two zero space of $e_d^n$ to prove \Cref{thm:alg_formulas}.

\section{Preliminaries}\label{sec:prelim}
Given a set $S$ and a number $k \in \NN$, let ${S \choose k}$ be the set of all subsets of $S$ of size $k$.
For every $n \in \NN$, let $[n] = \{1, \ldots, n\}$, and for every $m < n$, let $[m, n] = \{m, \ldots, n\}$.

We will let $\FF$ denote a field and $\chr(\FF)$ will denote the characteristic of the field.
We let $\FF[x_1, \ldots, x_n]$ denote the polynomial ring generated by $x_1, \ldots, x_n$ and $\FF(x_1, \ldots, x_n)$ denote the corresponding field of fractions.
Given a polynomial $f \in \FF[x_1, \ldots, x_n]$, we use $\deg(f)$ to denote its degree.
We say a polynomial is a form if it is homogeneous.

\subsection{The Elementary Symmetric Polynomials}\label{sec:esp}
In this section, we will describe some of the basic properties of the elementary symmetric polynomials, which are used throughout the paper.
Due to their simplicity, these polynomials possess many nice properties that we will use to simplify them and relate them to each other.
For example, given $a_1, \ldots, a_n, b_1, \ldots, b_m \in \FF$, we observe that
\begin{equation}\label{eq:esp_split}
\esym{n+m}{d}(a_1, \ldots, a_n, b_1, \ldots, b_m) = \sum_{k=0}^d\esym{n}{k}(a_1, \ldots, a_n)\esym{m}{d-k}(b_1, \ldots, b_m).
\end{equation}
We can also study their partial derivatives and observe that
\begin{equation}\label{eq:esp_partial_der}
\frac{\partial \esym{n}{d}}{\partial x_i}(x) = \esym{n}{d-1}(x) - x_i\frac{\partial \esym{n}{d-1}}{\partial x_i}(x) = \esym{n-1}{d-1}(x_1, \ldots, x_{i-1}, x_{i+1}, \ldots, x_n).
\end{equation}
We can then use this fact and a well-known property of homogeneous polynomials to conclude that
\begin{equation}\label{eq:esp_sum_partial_der}
\sum_{i=1}^nx_i\frac{\partial \esym{n}{d}}{\partial x_i}(x) = d \cdot \esym{n}{d}(x).
\end{equation}

Another important property of the elementary symmetric polynomials are Newton's identities.
We will denote $p_d^n(x_1, \ldots, x_n) = x_1^d + \dots + x_n^d$.
Then, Newton's identities are given by
\begin{equation}\label{eq:newton_identity}
d\esym{n}{d}(x_1, \ldots, x_n) = \sum_{k=1}^d(-1)^{k+1}p_k^n(x_1, \ldots, x_n)\esym{n}{d-k}(x_1, \ldots, x_n).
\end{equation}

\subsection{Algebraic Circuits}\label{subsec:alg_circ}
We define an algebraic circuit to be a directed acyclic graph with labeled vertices; vertices with no incoming edges are labeled with elements from the field $\FF$ or independent variables $x_1, \ldots, x_n$, whereas vertices with incoming edges are labeled with either $+$ for addition or $\times$ for multiplication.
We define the polynomial computed by a given vertex by: if it is labeled by an element of $\FF \cup \{x_1, \ldots, x_n\}$, then we use the label itself; if it is labeled by $+$, then it is the sum of the polynomials computed by its incoming vertices; and if it is labeled by $\times$, then it is the product of the polynomials computed by its incoming vertices.
We designate one vertex as the output vertex, and we say that the circuit computes the polynomial computed by this vertex.

If the underlying graph of a circuit is a tree, then we say that it is an algebraic formula.
We say that the depth of a formula is the depth of the underlying tree.
We say that a circuit is $\Sigma\Pi\Sigma$ if the depth of the formula is three, paths from root to leaf alternate between $+$ and $\times$, and the output node is labeled with $+$.
We then further say a circuit is $\Sigma^{[k]}\Pi\Sigma$ if the root node has at most $k$ incoming edges.

For more information on algebraic complexity, look at the surveys \cite{SY10} and \cite{Sap15}.

\subsection{Border Complexity}\label{subsec:border_comp}
\begin{definition}
A polynomial $f \in \FF[x_1, \ldots, x_n]$ is approximated by a polynomial $F \in \FF[\varepsilon][x_1, \ldots, x_n]$ if there exists an $N, M \in \NN$ and $Q_1, \ldots, Q_M \in \FF[x_1, \ldots, x_n]$ such that
\[
F = \varepsilon^N \cdot  f + \varepsilon^{N+1} \cdot \left(\sum_{i=1}^M\varepsilon^{i-1}Q_i\right).
\]
We then sometimes write $f \simeq F$.

Consider a complexity class $\C_\FF$ defined over a field $\FF$.
We say that a polynomial is in the border of the complexity class if it can be approximated by the class $\C_{\FF(\varepsilon)}$.
Let $\overline{\C_\FF}$ denote the set of all polynomials that are in the border of $\C_\FF$.
\end{definition}

\subsection{Algebraic Geometry}
For these results, we use some basic notions from algebraic geometry.
We define a variety to be the zero set of some polynomial equations; i.e., given some polynomials $f_1, \ldots, f_k \in \FF[x_1, \ldots, x_n]$, we will let $V(f_1, \ldots, f_k) \subseteq \FF^n$ define the variety of these equations, which we write as
\[
V(f_1, \ldots, f_k) = \{a \in \FF^n \mid f_1(a) = \dots = f_k(a) = 0\}.
\]

We further say that a variety $X \subseteq\FF^n$ is irreducible if, for every pair of varieties $Y, Z \subseteq \FF^n$ such that $X = Y \cup Z$, either $Y = X$ or $Z = X$.
Every variety $X$ has a unique irreducible decomposition $X = X_1 \cup \dots \cup X_\ell$ where $X_1, \ldots, X_\ell$ are irreducible and, for each $i \ne j$, $X_i \not\subseteq X_j$.
Observe that a vector space is an irreducible variety.
We will use some facts about irreducible varieties.
\begin{proposition}\label[proposition]{prop:irr_subset}
Let $X, Y_1, \ldots, Y_\ell \subseteq \FF^n$ be irreducible varieties such that $X \subseteq \cup_{i = 1}^\ell Y_i$, then there is an $i \in [\ell]$ such that $X \subseteq Y_i$.
\end{proposition}
We will often discuss the dimension of varieties, denoted $\dim(V)$, which is defined by the largest $k$ such that there is a chain of irreducible varieties $\emptyset \subsetneq V_0 \subsetneq \dots \subsetneq V_{k} \subseteq V$.
The main property of dimension that we will use is given in the following.
\begin{proposition}
Given varieties $V, U \subseteq \FF^n$ such that $V \subseteq U$, we have that $\dim V \le \dim U$.
\end{proposition}
This is even stronger when we discuss irreducible varieties.
\begin{proposition}\label[proposition]{prop:irred_neq_subset}
Let $X, Y \subseteq \FF^n$ be irreducible varieties such that $X \subsetneq Y$, then we have that $\dim X < \dim Y$.
\end{proposition}

For our purposes, we will focus on one particular variety, which we will call the order-two-zero space.
\begin{definition}
Given a polynomial $f \in \FF[x]$, we define the order-two-zero space of $f$, denoted $V_2(f)$ by
\[
V_2(f) = V\left(f, \frac{\partial f}{\partial x_1}, \ldots, \frac{\partial f}{\partial x_n}\right).
\]
\end{definition}
\begin{remark}\label{rem:sing_locus}
This idea of the order-$2$-zero space is utilized in \cite{Kum19} and \cite{CKSV22} without a name.
Then, in \cite{KV22}, it is introduced as the ``singular locus" of a polynomial, denoted $\operatorname{sing}(f)$, which is a well-known object in algebraic geometry, in a nod to the notation of \cite{vzG87}.
For this paper, we have decided to change this notation, as this definition of singular locus does not precisely align with the algebraic-geometric definition.
Although it is true that, when $f$ is square-free, this definition exactly characterizes the singular locus ($V_2(f) = \operatorname{sing}(f)$), this is not true in general.
This follows from the fact that the singular locus of a polynomial is a property of its corresponding hypersurface, a purely geometric object, while the order-$2$-zero space is a property of the polynomial, itself.
Specifically, if we consider the power polynomial $p_d^n = x_1^d + \dots + x_n^d$, it is well-known that $\sing (p_d^n) = \{0\}$ (see Example 10.21 of \cite{Gat21}), but, if we consider $\operatorname{char}(\FF) = q \ne 0$, we trivially observe that $V_2(p_q^n) = V(p_q^n)$.
We find that letting this definition differ from its natural geometric definition is confusing, so we have decided to change it.
\end{remark}

\section{Order-$2$ zero space of the elementary symmetric polynomials}\label{sec:order_two_zeros}
In this section, we will focus on extending the formula complexity lower bound of the elementary symmetric polynomials from \cite{CKSV22}.
Specifically, we will show that these results extend when we consider fields of positive characteristic, which we do by analyzing the order-$2$-zero set of the elementary symmetric polynomials, as this is the only part of the proof that uses the characteristic of the field.
Upon careful inspection of \cite{CKSV22}, we can state the main result we rely on in the following lemma. For the sake of completeness, we prove this lemma in the appendix in Section \ref{sec:formula_lower_bounds}.
\begin{lemma}[\cite{CKSV22}]\label{lem:formula_lower_bound}
Suppose $f \in \FF[x_1, \ldots, x_n]$ is a polynomial of degree $d \ge 3$ that can be computed by a formula $\Phi$ of size $s$. Then,
\[
s \ge \frac{d}{6}(n - \dim V_2(f))
\]
\end{lemma}
Then, in \cite{CKSV22}, the proof is completed through the following proposition.
\begin{proposition}[\cite{MZ17}, Lemma 12 of \cite{LMP19}, Lemma 5.2 of \cite{CKSV22}]\label{clm:char_zero_esp_sing}
Let $\chr(\FF) = 0$. If $2 \le d \le n$, then $\dim V_2(\esym{n}{d}) = d - 2$.

In particular, we have that $V_2(\esym{n}{d}) = \bigcup_{I \in {[n] \choose d-2}}\{a \in \FF^n \mid \forall\ i \notin I, a_i = 0\}$.
\end{proposition}
One should notice that a formula complexity lower bound of $\Omega(d(n-d))$ on $\esym{n}{d}$ for $\chr(\FF) = 0$ immediately follows from this.
We will spend the rest of this section studying this result over fields of positive characteristic.

One may first ask whether the result from \Cref{clm:char_zero_esp_sing} can be extended to fields of positive characteristic as stated.
From the following example, we can see that this is not the case.
\begin{example}
Consider the field $\mathbb{F}_2 = \{0, 1\}$ (or even its algebraic closure $\overline{\FF_2}$).
We consider the set $V_2(\esym{5}{2})$, and we observe that $\{(\alpha, \alpha, \alpha, \alpha, \alpha) \mid \alpha \in \overline{\FF_2}\} \subseteq V_2(\esym{5}{2})$, implying that $\dim V_2(\esym{5}{2}) \ge 1$.
To see this, consider an arbitrary $\alpha \in \overline{\FF_2}$, and notice that
\[
\esym{5}{2}(\alpha, \alpha, \alpha, \alpha, \alpha) = {5 \choose 2}\alpha^2 = 10\alpha^2 = 0,\ \frac{\partial \esym{5}{2}}{\partial x_i}(\alpha, \alpha, \alpha, \alpha, \alpha) = \esym{4}{1}(\alpha, \alpha, \alpha, \alpha) = 4\alpha = 0.
\]
\end{example}

Clearly, we cannot hope to get an upper bound of $(d-2)$ for arbitrary choices of field characteristic, number of variables, and degree.
But, if we instead turn our attention to an upper bound of $(d-1)$, this is possible.
Luckily, for the sake of asymptotic bounds, the difference between $(d-1)$ and $(d-2)$ is not important.

Before providing the proof, we will begin by analyzing some of the smaller cases, which we hope will provide insight into the general case.
To begin with, our earlier analysis seems to indicate that $\dim V_2(e_2^n) \le 1$.
This is quite simple to see; we begin by noticing that
\[
\frac{\partial e_2^n}{\partial x_n} = x_1 + \dots + x_{n-1}.
\]
From this observation, we also notice that $\frac{\partial e_2^n}{\partial x_1} - \frac{\partial e_2^n}{\partial x_n} = x_n - x_1$.
This clearly generalizes to all pairs of variables.
Thus, we conclude that any point in $V_2(e_2^n)$ has the form $(\alpha, \ldots, \alpha)$, which gives us the desired conclusion.

From our argument for the smaller case, we may hope that we can generalize it and arrive at an inductive argument.
Let us do that to show that $\dim V_2(e_3^n) \le 2$.
Using a similar argument to previous, we note that
\[
\frac{\partial e_3^n}{\partial x_1} - \frac{\partial e_3^n}{\partial x_n} = e_2^{n-1}(x_2, \ldots, x_n) - e_2^{n-1}(x_1, \ldots, x_{n-1}) = (x_n - x_1)\frac{\partial^2 e_3^n}{\partial x_1 \partial x_n}.
\]
This is again true for any arbitrary pair of variables we choose.
In essence, we have shown that, for any pair of variables, their value is either equal, or the second order derivative in terms of these variables is zero.
Without loss of generality, we can assume that $x_1 = \dots = x_\ell$, which are different from $x_{\ell+1}, \ldots, x_n$.
Then, notice that
\[
0 = \frac{\partial^2 e_3^n}{\partial x_1 \partial x_j} = \frac{\partial e_2^{n-\ell}}{\partial x_j}(x_{\ell+1}, \ldots, x_n) + (\ell - 1)x_1,
\]
for $j \ge \ell+1$.
We had hoped to use the second-order partial derivatives to reduce the problem to analyzing the order-two-zero set of an elementary symmetric polynomial of lower degree, but we now have an extra additive constant.
To properly generalize our inductive argument, we will need to fix our inductive hypothesis to be, instead of $\frac{\partial e_2^n}{\partial x_i} = 0$, the equation $\alpha_1\cdot \frac{\partial e_2^n}{\partial x_i} + \alpha_0 = 0$.
We note that this generalizes to $\sum_{k=0}^d\alpha_k\frac{\partial e_{k+1}^n}{\partial x_i} = 0$.
We can now successfully apply induction and finish the proof.
We will now formalize this approach.

\begin{lemma}\label[lemma]{lem:esp_sing_nonzero}
Over any field, if $1 \le d \le n$, then $\dim V_2(\esym{n}{d}) \le d-1$.
\end{lemma}
\begin{proof}
We will show that all points in $V_2(\esym{n}{d})$ have at most $d-1$ distinct values in their coordinates.
To define this rigorously, we first define vector spaces corresponding to each partition of $[n]$.
Given a set $S$, we say that $\pi$ in the power set of $S$ is a partition of $S$ of size $k$ if $|\pi| = k$, $\emptyset \notin \pi$, $\bigcup_{U \in \pi}U = S$, and the sets in $\pi$ are pairwise disjoint.
We say that the set of all such $\pi$ is $\P_k(S)$.
Given a partition $\pi \in \P_k([n])$, we define the vector space $P_\pi$ formed by $\pi$ to be the set of all points such that, for every pair of indices that are in the same set in the partition, they also always have the same value in the space.
Rigorously, we say that
\begin{equation}\label{eq:perm_set}
P_\pi = \{a \in \FF^n \mid \forall\ S \in \pi,\ \forall\ i, j \in S,\ a_i = a_j\}.
\end{equation}
Then, we can define the set of all points that can be split into $k$ groups with the same value by
\[
S_k = \bigcup_{\pi \in \P_k([n])}P_\pi.
\]
We will spend the rest of the proof showing that $V_2(\esym{n}{d}) \subseteq S_{d-1}$.
From this, the claim immediately follows, as $\dim S_k = k$.

To prove this result, we will inductively prove the following claim.
\begin{claim}
For $d \ge 0$ and $m > d$, let $\alpha_0, \ldots, \alpha_d \in \FF$ be such that $\alpha_d = 1$.
Then, we have that
\[
\left\{a \in \FF^m\ \middle\vert\ \forall\ i \in [m],\ \sum_{k = 0}^d\alpha_k\frac{\partial e_{k+1}^m}{\partial x_i}(a) = 0\right\} \subseteq S_d.
\]
\end{claim}
We can apply the above claim, setting $d$ to $(d-1)$, $\alpha_{d-1}$ to $1$, and $\alpha_{d-2} = \dots = \alpha_0 = 0$, to prove the lemma.
We will now focus on proving the claim, which we will do by induction on $d$.
Notice that the case of $d = 0$ is obvious, as $\alpha_0 = 1$, so the set is empty.
Now, consider the claim for $d$, using the inductive hypotheses for $(d-1)$.
Let $a = (a_{1}, \ldots, a_m)$ be such that, for every $i \in [m]$, $\sum_{k=0}^d\alpha_k\frac{\partial e_{k+1}^m}{\partial x_i}(a_1, \ldots, a_m) = 0$.
Observe that
\[
\begin{split}
0 & = \sum_{i=0}^d\alpha_i\esym{m-1}{i}(a_1, \ldots, a_{m-1}) - \sum_{i=0}^d\alpha_i\esym{m-1}{i}(a_{2}, \ldots, a_{m}) = (a_{1} - a_m)\sum_{i=1}^d\alpha_{i}\esym{m-2}{i-1}(a_{2}, \ldots, a_{m-1}).
\end{split}
\]
Hence, we have that either $a_{1} = a_m$ or $\sum_{i=1}^d\alpha_{i}\esym{m-2}{i-1}(a_{2}, \ldots, a_{m-1}) = 0$.
Observe that we can easily extend this argument to any pair of indices.

Now, let $I = \{k \in [m-1] \mid a_k = a_m\}$.
Without loss of generality, assume that $[m-1] \setminus I = [\ell]$.
Observe that if $\ell < d$, then the claim trivially follows, so we assume that $\ell \ge d$.
Consider some $k \in [\ell]$, and, using the fact that $a_i = a_m$ for each $i \in I$ and $a_k \ne a_m$, observe that
\[
\begin{split}
0 &= \sum_{i=0}^{d-1}\alpha_{i+1}\frac{\partial e_{i+1}^{m-1}}{\partial x_k}(a_1, \ldots, a_{m-1}) = \sum_{i=0}^{d-1}\alpha_{i+1}\sum_{j=0}^i{|I| \choose i-j}a_m^{i-j}\frac{\partial e_{j+1}^\ell}{\partial x_k}(a_1, \ldots, a_\ell)\\
&= \sum_{j=0}^{d-1}\left(\sum_{i=j}^{d-1}\alpha_{i+1}{|I| \choose i-j}a_m^{i-j}\right)\frac{\partial e_{j+1}^\ell}{\partial x_k}(a_1, \ldots, a_\ell) = \sum_{j=0}^{d-1}\alpha_j'\frac{\partial e_{j+1}^\ell}{\partial x_k}(a_1, \ldots, a_\ell),
\end{split}
\]
where $\alpha_i' \in \FF$.
Observe that $\alpha_{d-1}' = \alpha_d \cdot {|I| \choose 0} = 1$.
We then can use the inductive hypothesis to conclude that $(a_1, \ldots, a_\ell) \in S_{d-1}$.
Because $a_{\ell+1} = \dots = a_m$, we conclude that $a \in S_d$.
\end{proof}

With the above proof, we have now extended the results of \cite{CKSV22} to fields of positive characteristic.
We will now provide a condition for when the dimension is $(d-2)$ and when the answer is $(d-1)$.
\begin{theorem}
Let $\operatorname{char}(\FF) = p \ne 0$, $\FF$ be an infinite field, and $2 \le d \le n$.
If $n - d + 1 \equiv 0 \mod p^{\lfloor \log_p d \rfloor+1}$, then $\dim V_2(e_d^n) = d-1$.
Otherwise, we have that $\dim V_2(e_d^n) = d-2$.
\end{theorem}
Because the proof of this is not necessary for extending the results of \cite{CKSV22}, we have placed its proof in \Cref{sec:zero_space_anal} of the appendix.

\section{The Symmetric Model}\label{sec:symmetric}
In this section, we will study a computational model defined from the elementary symmetric polynomials.
Specifically, we will define Sym to represent the set of homogeneous polynomials of some degree, say $d$, that can be written as $\esym{m}{d}(L_1, \ldots, L_m)$ for some linear (homogeneous) polynomials $L_1, \ldots, L_m \in \FF[x_1, \ldots, x_n]$.
We will then say that $\Sigma^{[k]}\text{Sym}$ are all polynomials that can be written as a linear combination of $k$ elements of Sym.
We will use $\overline{\text{Sym}}$ and $\overline{\Sigma^{[k]}\text{Sym}}$ to denote the border versions of these classes.

This section will focus on proving \Cref{thm:main}.
To prove this theorem, it is sufficient to study the symmetric classes, as is shown by the following theorem.
\begin{theorem}\label{thm:depth_three_to_sym}
If $f \in \FF[x_1, \ldots, x_n]$ is homogeneous degree $d$ and can be represented using  $\overline{\Sigma^{[k]}\Pi\Sigma}$, then we can represent $f$ by $\overline{\Sigma^{[k]}\text{Sym}}$.
\end{theorem}
The proof of this largely relies on \Cref{eq:prod_to_esp}.
We will include the proof in \Cref{sec:reduction} of the appendix.

In this section, we will focus on proving:
\begin{theorem}\label{thm:esm_rep}
If $\operatorname{char}(\FF) \ne 0$, then there exists a polynomial in at most $n$ variables that can only be computed by $\overline{\Sigma^{[k]}\operatorname{Sym}}$ for $k = \Omega(n)$.
In other words, there is a homogeneous polynomial $f \in \FF[x_1, \ldots, x_n]$ of degree $d$ such that, for every linear $L_1^{(1)}, \ldots, L_{m_1}^{(1)}, L_1^{(2)}, \ldots, L_{m_k}^{(k)} \in \FF(\varepsilon)[x_1, \ldots, x_n]$ and constants $c_1, \ldots, c_k \in \FF(\varepsilon)$,
\begin{equation}\label{eq:esp_sum_form}
f \not\simeq \sum_{i=1}^kc_i\esym{m_i}{d}(L_1^{(i)}, \ldots, L_{m_i}^{(i)}),
\end{equation}
for $k = o(n)$.
\end{theorem}

The model of computation, Sym, was introduced and studied in \cite{Shp02}, where it was shown that, in fields of characteristic zero, it is complete.
The proof relied on the fact that, for every homogeneous polynomal $f$ of degree $d$, there are linear forms $L_1, \ldots, L_m$ such that $f = L_1^d + \ldots + L_m^d$.
This is precisely the step in the proof that fails in positive characteristic.
Recall that, if $\chr(\FF) = p$, then $(a + b)^p = a^p + b^p$.
Thus, if we select some $d \ge p$, we observe that each monomial in some $L^d$ for a linear $L \in \FF[x_1, \ldots, x_n]$ must be divisible by some $x_i^p$.
This is the fact that we will use for our proofs (we will in-fact use the weaker fact that it cannot contain multilinear monomials).

For now, it will be useful to state the result proved in \cite{Shp02} but in slightly more generality than was originally done.
This is based on Lemma 2.4 of \cite{Shp02}.
\begin{lemma}\label{lem:powers_free}
Assume that $\FF$ is algebraically closed (or simply that $z^d + 1$ is fully reducible in $\FF$). If $f \in \FF[x_1, \ldots, x_n]$ is a degree $d$ homogeneous polynomial that is in Sym and $L \in \FF[x_1, \ldots, x_n]$ is linear, then $f + L^d$ is in Sym.
This is also true if we replace Sym with $\overline{\operatorname{Sym}}$.
\end{lemma}
This lemma is proved in \Cref{sec:comp_sym_model}.

In this section, for simplicity, we will focus on proving \cref{thm:main} in the case when $\chr(\FF) = 2$.
We do this because we believe that this case shows off all of the tools needed for the proof without having to deal with the added complexity of an arbitrary field characteristic.
We include the proof for arbitrary field characteristic in \Cref{sec:sym_model_omit} of the appendix.
Hence, for the rest of this section, we will assume that $\chr(\FF) = 2$.

We will now take a moment to consider the choice of counterexample polynomials.
As previously stated, we will focus on multilinear polynomials.
We begin by considering degree-two polynomials.
\begin{claim}\label{clm:char_two_deg_two}
For every degree-two homogeneous polynomial $f \in \FF[x_1, \ldots, x_n]$, there are linear $L_1, \ldots, L_m \in \FF[x_1, \ldots, x_n]$ such that $\esym{m}{2}(L_1, \ldots, L_m) = f$ and $\esym{m}{1}(L_1, \ldots, L_m) = 0$.\footnote{This does not extend as-stated to $\chr(\FF) > 2$.}
\end{claim}
We prove this claim in \Cref{sec:comp_sym_model} of the appendix.
We thus turn our attention to degree three, which will give us what we desire.
\begin{claim}
If $\chr(\FF) = 2$, then the polynomial $f = \sum_{i=1}^\ell x_iy_iz_i$ is not computable by $\overline{\Sigma^{k}Sym}$ for $k \le \ell - 1$.\footnote{For a field of characteristic $p$, we use the inner-product polynomial of degree $(p+1)$.}
\end{claim}

We will spend the rest of this section on this claim.
Our proof will rely on studying a property of polynomials similar to the strength of a polynomial, which we will call the multilinear strength.
\begin{definition}[Multilinear Strength]
Given a polynomial $f \in \FF[x_1, \ldots, x_n]$, let $\operatorname{MULT}(f)$ be the multilinear part of $f$, i.e., the sum of all multilinear monomials in $f$.
Then, we say that the \emph{multlinear strength} of a multilinear polynomial $f$, denoted $s_{\operatorname{mult}}(f)$, is the smallest $k$ such that there are reducible polynomials $g_1, \ldots, g_k \in \FF[x_1, \ldots, x_n]$ where $f = \operatorname{MULT}(g_1 + \dots + g_k)$.
We define the \emph{approximate multilinear strength} of $f$, denoted $\overline{s_{\operatorname{mult}}}(f)$, as the infimum over the multilinear strengths of polynomials approximating $f$.
\end{definition}
\begin{example}
Consider the polynomial $f = x^2 + xy + y^2 + xz$.
Then, we have that $\operatorname{MULT}(f) = xy + xz$.
Further notice that we can write $\operatorname{MULT}(f) = x(y+z)$, so we have that $s_{\operatorname{mult}}(f) = 1$.
\end{example}

The proof will then follow by showing that polynomials computed by $\overline{\Sigma^{[k]}\operatorname{Sym}}$ have approximate multilinear strength at most $k$, but $\sum_{i=1}^\ell x_iy_iz_i$ has approximate multilinear strength exactly $\ell$.
We begin by proving the upper bound on approximate multilinear strength.
This is done with a simple application of the Newton identities.
\begin{claim}\label{clm:char_two_esp_to_reduc}
For every degree-three homogeneous polynomial $f \in \FF[x_1, \ldots, x_n]$, $f$ is in $\Sigma^{[k]}\operatorname{Sym}$ if and only if\footnote{We only get one direction if $\chr(\FF) > 2$.} there are reducible homogeneous degree-three polynomials $g_1, \ldots, g_k \in \FF[x_1, \ldots, x_n]$ and linear $q_1, \ldots, q_\ell \in \FF[x_1, \ldots, x_n]$ such that
\[
f = g_1 + \ldots + g_k + q_1^3 + \dots + q_\ell^3.
\]
\end{claim}
\begin{proof}
To prove the claim, it is sufficient to prove it for $k = 1$.
Suppose that $f$ is in Sym.
Let $L_1, \ldots, L_m \in \FF[x_1, \ldots, x_n]$ be linear such that $\esym{m}{3}(L_1, \ldots, L_m) = f$.
Then, by the Newton identities (\ref{eq:newton_identity}) and the fact that $p_2^m = (\esym{m}{1})^2$, we observe that
\begin{align*}
f &= \esym{m}{2}(L_1, \ldots, L_m)p_1^m(L_1, \ldots, L_m) + \esym{m}{1}(L_1, \ldots, L_m)p_2^m(L_1, \ldots, L_m) + p_3^m(L_1, \ldots, L_m) \\
&= \esym{m}{2}(L_1, \ldots, L_m)e_1^m(L_1, \ldots, L_m) + (\esym{m}{1}(L_1, \ldots, L_m))^3 + L_1^3 + \dots + L_m^3.
\end{align*}

For the other direction, let $g \in \FF[x_1, \ldots, x_n]$ be a reducible homogeneous degree-three polynomial and $q_1, \ldots, q_k \in \FF[x_1, \ldots, x_n]$ be linear.
Let $g_1 \in \FF[x_1, \ldots, x_n]$ be homogeneous and degree two and $g_2 \in \FF[x_1, \ldots, x_n]$ be linear such that $g = g_1g_2$.
By Claim \ref{clm:char_two_deg_two}, let $L_1, \ldots, L_m \in \FF[x_1, \ldots, x_n]$ be linear such that $g_1 = \esym{m}{2}(L_1, \ldots, L_m)$ and $0 = \esym{m}{1}(L_1, \ldots, L_m)$. Then, observe that
\begin{align*}
\esym{m}{3}(L_1, \ldots, L_m, g_2) &= \esym{m}{2}(L_1, \ldots, L_m, g_2) \cdot \esym{m}{1}(L_1, \ldots, L_m, g_2) + (\esym{m}{1}(L_1, \ldots, L_m, g_2))^3 \\
&+ L_1^3 + \dots + L_m^3 + g_2^3 \\
&= g_1g_2 + (\esym{m}{1}(L_1, \ldots, L_m, g_2))^3 + L_1^3 + \dots + L_m^3 + g_2^3.
\end{align*}
Then, the proof is completed by Lemma \ref{lem:powers_free}.
\end{proof}

\begin{corollary}\label[corollary]{cor:s_mult}
If a degree-three homogeneous polynomial $f \in \FF[x_1, \ldots, x_n]$ is in $\Sigma^{[k]}\operatorname{Sym}$, then $s_{\operatorname{mult}}(f) \le k$.
If $f$ is in $\overline{\Sigma^{[k]}\operatorname{Sym}}$, then $\overline{s_{\operatorname{mult}}}(f) \le k$.
\end{corollary}

Now, we can focus on proving the theorem.
Instead of the entire theorem, we will focus on the case when the fan-in is one, and we will prove the general case in \Cref{sec:sym_model_omit}.
We should notice that this case is enough to show that \cite{Kum20} does not extend to positive characteristics.

\begin{claim}\label{clm:char_two_sym}
The polynomial $f = x_1y_1z_1 + x_2y_2z_2 \in \FF[x_1, \ldots, x_6]$ is such that $\overline{s_{\operatorname{mult}}}(f) = 2$.
\end{claim}
\begin{proof}
For simplicity, we will write $f = x_1x_2x_3 + x_4x_5x_6$.
Suppose that $\overline{s_{\operatorname{mult}}}(f) = 1$ for the sake of reaching a contradiction.
Let $g \in \FF[\varepsilon][x_1, \ldots, x_n]$ be reducible and such that $f \simeq g$.
We will now derive a contradiction.
We will focus now on the non-border setting, and then we will comment at the end that this proof also works for the border.

We will denote the coefficient of an arbitrary multilinear monomial ($x_ix_jx_k$) of $g$ by $c_{ijk}$.
Because $g$ is reducible, we will split it into a linear part, denoted $g_1$, and a quadratic part, denoted $g_2$, given by
\[
g_1 = \sum_{i=1}^6a_ix_i,\ g_2 = \sum_{i=1}^6\sum_{j=i}^6b_{ij}x_ix_j,
\]
where $g = g_1g_2$.
Notice that $c_{ijk} = a_ib_{jk} + a_jb_{ik} + a_kb_{ij}$.

Our goal will be to use the fact that $c_{123} \ne 0$ and $c_{456} \ne 0$ (meaning that $c_{123}c_{456} \ne 0$) to show that some other $c_{ijk} \ne 0$, a contradiction.
We will consider these $c_{ijk}$ values as polynomials in the coefficients of $g_1$ and $g_2$.

Notice that that there must be a coefficient (in terms of the $a$'s and $b$'s) within $c_{123}c_{456}$ that is non-zero, say $a_1b_{23}a_4b_{56}$, for example.
This monomial also has a nonzero coefficient in $c_{234}c_{156}$ (which corresponds to a different variable partition).
In fact, if we consider all possible partitioning of the variables, this is the only other partition that has this coefficient nonzero.
Our strategy will then be to take combinations of these pairs formed from partitions, which will form a `metapolynomial' in the coefficients of $f$ that vanishes.

We will state this more formally to show that this indeed happens.
Let $\mathcal{P}_3[6]$ represent all of the ways to partition $[6]$ into two sets of size three.\footnote{Note that the notation for partitions is slightly different from the notation used in previous sections.}
Then, we will show that
\begin{equation}\label{eq:char_two_sym}
\sum_{\mathcal{I} \in \mathcal{P}_3[6]}\prod_{\{i, j, k\} \in \mathcal{I}}c_{ijk} \equiv 0 \mod 2.
\end{equation}
Observe that this equality proves the claim, as $c_{123} = c_{456} = 1$, implying that there is another partition $\{i, j, k\}, \{i', j', k'\}$ such that $c_{ijk}c_{i'j'k'} \ne 0$, which is a contradiction.

We now prove the equality.
Consider (\ref{eq:char_two_sym}) as a polynomial in the indeterminants given by $a_i$ and $b_{jk}$.
Observe that the monomials in this equation can be written in the form $a_ib_{jk}a_{i'}b_{j'k'}$, where $\{i, j, k, i', j', k'\} = [6]$.
But this monomial only appears in $c_{ijk}c_{i'j'k'}$ and $c_{i'jk}c_{ij'k'}$.
Further notice that the corresponding coefficient is one in these terms.
This lets us conclude that
\[
\sum_{\mathcal{I} \in \mathcal{P}_3[6]}\prod_{\{i, j, k\} \in \mathcal{I}}c_{ijk} = 2\sum_{\{i, i'\} \subseteq [6]}\sum_{\{\{j, k\}, \{j', k'\}\} \in \mathcal{P}_2([6] \setminus \{i, i'\})}a_ia_{i'}b_{jk}b_{j'k'}.
\]
We conclude with the observation that this creates an identity in the coefficients of our underlying polynomials.
The proof thus works in the border setting.
\end{proof}

\section*{Acknowledgements}
Thank you to my advisor Srikanth Srinivasan for encouraging and helping me throughout this entire line of research and steering me in the right directions through regular chats.
A lot of the ideas in this paper came from our long discussions.
Also, thank you for helping during the process of writing this paper, where your extensive knowledge of the field was paramount.
Further, thank you to Théo Fabris for being involved in some of these aforementioned discussions.

\bibliography{references}

\newpage
\appendix

\section{Exact Zero-Space Analysis}\label{sec:zero_space_anal}
This section in the appendix will focus on proving \Cref{lem:order_two_esp}.
We restate it here to keep the argument contained.
\begin{lemma}
Let $\operatorname{char}(\FF) = p \ne 0$, $\FF$ be an infinite field, and $2 \le d \le n$.
If $n - d + 1 \equiv 0 \mod p^{\lfloor \log_p d \rfloor+1}$, then $\dim V_2(e_d^n) = d-1$.
Otherwise, we have that $\dim V_2(e_d^n) = d-2$.
\end{lemma}

The proof of this will largely rely on Lucas's theorem.
We recall the statement here.
\begin{theorem}[Lucas's Theorem, see \cite{Luc78}, \cite{Mes14}]
Let $p \in \NN$ be a prime and $a, b \in \NN$ be numbers. We then write $a$ and $b$ by their unique base-$p$ expansion, namely, $a = \sum_{i=0}^ka_i \cdot p^i$ and $b = \sum_{i = 0}^kb_i \cdot p^i$. Then, we have that ${a \choose b} \ne 0$ if and only if, for each $i \in [0,k]$, we have that $b_i \le a_i$.
\end{theorem}
\begin{note}
Consider a number $a$ and its base $p$ expansion $a = \sum_{i=0}^ka_i \cdot p^i$.
We will often call the $a_i$ values the $p$-digits of $a$.
We may further refer to $a_i$ as the $i$th $p$-digit of $a$.
\end{note}
Our strategy we rely on analyzing many such combinations and determining if they are zero or not, so will introduce the notation $a \le_p b$ to indicate that all of the $p$-digits of $a$ are less than or equal to the $p$-digits of $b$.
Hence, $a \le_p b$ if and only if ${b \choose a} \ne 0$.

We will now give an overview of our proof strategy.
We begin by recalling some notation and facts from the proof of \Cref{lem:esp_sing_nonzero}.
Let $\P_k([n])$ be the set of all partitions of $[n]$ into $k$ groups.
Given a $\pi \in \P_k([n])$, we define the vector space $P_\pi$ as in (\ref{eq:perm_set}).
Recall that the strategy for proving \Cref{lem:esp_sing_nonzero} was to show that $V_2(e_d^n) \subseteq \bigcup_{\pi \in \P_{d-1}([n])}P_\pi$, which implied that $\dim V_2(e_d^n) \le d-1$, as the dimension of each $P_\pi$ is exactly $(d-1)$.
Further recall that $\dim V_2(e_d^n) \ge d-2$, as all elements with $(n-d+2)$ coordinates equal to zero are a part of this set.

We now present a simple strategy to determine whether $\dim V_2(e_d^n)$ is $(d-2)$ or $(d-1)$.
First, we notice that if we can find a $\pi$ such that $P_\pi \subseteq V_2(e_d^n)$, then we get a lower bound of $(d-1)$.
On the other hand, in order to get an upper bound of $(d-2)$, observe that it is sufficient to show that, for every $\pi$, $V_2(e_d^n) \cap P_\pi \ne P_\pi$.
This follows from fundamental facts in algebraic geometry: if we consider an irreducible component $X \subseteq V_2(e_d^n)$, then there is a $\pi$ such that $X \subseteq P_\pi$ (see \Cref{prop:irr_subset}), meaning that $X \subsetneq P_\pi$ by assumption, but this implies that $\dim X < \dim P_\pi = d-1$ (see \Cref{prop:irred_neq_subset}).

Now, suppose we had some point $a \in V_2(e_d^n)$, and let $\pi \in \P_{d-1}([n])$ be such that $a \in P_\pi$.
Define $\beta_1, \ldots, \beta_{d-1}$ to be the size of the sets in $\pi$, observing that $\beta_1 + \dots + \beta_{d-1} = n$.
Let $a_1, \ldots, a_{d-1}$ be the unique coordinates in $a$ so that $a_i$ is repeated $\beta_i$ times.
We can now consider the equations of $V_2(e_d^n)$:
\begin{equation}\label{eq:esm_zero}
0 = e_d^n(a) = \sum_{\begin{array}{c}
    k_1, \ldots, k_{d-1}\\
    k_i \le \beta_i\\
    k_1 + \dots + k_{d-1} = d
\end{array}}{\beta_1 \choose k_1} \dots {\beta_{d-1} \choose k_{d-1}}a_1^{k_1} \dots a_{d-1}^{k_{d-1}},
\end{equation}
and
\begin{equation}\label{eq:esm_deriv_zero}
0 = \frac{\partial e_d^n}{\partial x_i}(a) = \sum_{\begin{array}{c}
    k_1, \ldots, k_{d-1}\\
    k_i \le \beta_i - 1\\
    k_j \le \beta_j\\
    k_1 + \dots + k_{d-1} = d-1
\end{array}}{\beta_1 \choose k_1} \dots {\beta_i - 1 \choose k_i} \dots {\beta_{d-1} \choose k_{d-1}}a_1^{k_1} \dots a_{d-1}^{k_{d-1}}.
\end{equation}
Thus, to prove the lower bound, it suffices to prove that all of these products of combinations are zero; whereas, to prove an upper bound, one must find one of these products that are nonzero.

We are now ready to prove the lower bound part of the theorem.
\begin{lemma}
Let $\chr(\FF) = p \ne 0$ and $2 \le d \le n$. If $n - d + 1 \equiv 0 \mod p^{\lfloor\log_pd\rfloor + 1}$, then $\dim V_2(\esym{n}{d}) \ge d-1$.
\end{lemma}
\begin{proof}
We will begin by selecting a partition.
Consider the partition $\pi = \{\{1\}, \{2\}, \ldots, \{d-2\}, \{d-1, d, \ldots, n\}\}$.
We will prove that $P_\pi \subseteq V_2(e_d^n)$.

Using (\ref{eq:esm_zero}) and (\ref{eq:esm_deriv_zero}), we see that it suffices to show that ${n-d+2 \choose i} \equiv 0 \mod p$ for $i \in [2, d]$ and ${n-d+1 \choose i} \equiv 0 \mod p$ for $i \in [d-1]$.
To show this, notice that the $0$th to ($\lfloor \log_pd\rfloor$)th $p$-digits are zero in the unique base-$p$ expansion of $(n-d+1)$ by hypothesis.
Hence, the first to ($\lfloor \log_pd\rfloor$)th $p$-digits are zero and the $0$th $p$-digit is $1$ for $(n-d+2)$.
But every number in $[d-1]$ has one of the $0$th to ($\lfloor \log_pd\rfloor$)th $p$-digits nonzero and every number in $[2, d]$ has either one of the first to ($\lfloor \log_pd\rfloor$) $p$-digits nonzero or the $0$th digit greater than one.
The proof is complete by Lucas's Theorem.
\end{proof}

We then prove the other direction.
\begin{lemma}
Let $\operatorname{char}(\FF) = p \ne 0$, $\FF$ be an infinite field, and $2 \le d \le n$.
If $n - d + 1 \not\equiv 0 \mod p^{\lfloor \log_pd\rfloor+1}$, then $\dim V_2(e_d^n) \le d-2$.
\end{lemma}
\begin{proof}
Consider an arbitrary $\pi \in \P_{d-1}([n])$, and we will show that $V_2(e_d^n) \cap P_\pi \ne P_\pi$.
Let $\beta_1, \ldots, \beta_{d-1}$ be the corresponding group sizes.
To prove the theorem, we will show that we can satisfy one of the two following goals:
\begin{itemize}
    \item There are $k_i \le_p \beta_i$ such that $k_1 + \dots + k_{d-1} = d$ or
    \item We can set some $\beta_i \leftarrow \beta_i - 1$ so that there are $k_j \le_p \beta_j$ such that $k_1 + \dots + k_{d-1} = d-1$.
\end{itemize}
One should observe that, due to Lucas's Theorem, satisfying one of these criteria is enough to prove the result. We will split our analysis into cases.

\emph{Case $1$:}
We will assume that there is a $\beta_i$ such that the $0$th $p$-digit is zero.
If there are multiple options, we will take $\beta_i$ to be the one with the highest least-significant $p$-digit.
We will set $\beta_i \leftarrow \beta_i - 1$ and complete the second goal.
We will rely on the following simple observation: there is a $k > 0$ such that the $0$th to $(k-1)$th $p$-digits of $\beta_i$ are equal to $(p-1)$ and every $\beta_j$ is such that one of the $0$th to $k$th $p$-digits are non-zero.
From this information, we can fulfill the second goal.

\emph{Case $2$:}
We now assume that there is a $\beta_i$ such that the $0$th $p$-digit is greater than $1$.
Notice that we can set $k_i = 2$ and $k_j = 1$ for $j \ne i$ to fulfil the first goal.

\emph{Case $3$:}
In the final case, we have that, for each $\beta_i$, the $0$th $p$-digit of it is exactly one.
Observe that $(\beta_1 - 1) + \dots + (\beta_{d-1} - 1) = n - d + 1 \not\equiv 0 \mod p^{\lfloor \log_pd \rfloor +1}$.
Thus, there must be a $\beta_i$, where one of its first to ($\lfloor \log_pd \rfloor$)th $p$-digit is non-zero.
Thus, we can achieve the first goal by using this $p$-digit along with the $0$th $p$-digits.
\end{proof}

\section{Computation in the Symmetric Model}\label{sec:comp_sym_model}
\begin{proof}[Proof of \Cref{lem:powers_free}]
Let $\omega_1, \ldots, \omega_d \in \FF$ be all of the solutions to $z^d + 1=0$ (counted with multiplicities).
Then, we know that $\esym{d}{d}(-\omega_1, \ldots, -\omega_d) = 1$ and $\esym{d}{k}(-\omega_1, \ldots, -\omega_d) = 0$ for every $k \in [d-1]$, as
\[
z^d + 1 = \prod_{i=1}^d(z - \omega_i) = \sum_{k=0}^dz^{d-k}\esym{d}{k}(-\omega_1, \ldots, -\omega_d).
\]

Now, let $L_1, \ldots, L_m \in \FF[x_1, \ldots, x_n]$ be linear such that $\esym{m}{d}(L_1, \ldots, L_m) = f$.
Then, observe that, using (\ref{eq:esp_split}),
\begin{align*}
\esym{m+d}{d}(L_1, \ldots, L_m, -\omega_1q, \ldots, -\omega_dq) &= \sum_{i=0}^d\esym{m}{i}(L_1, \ldots, L_m)\esym{d}{d-i}(-\omega_1q, \ldots, -\omega_dq) = f + q^d.
\end{align*}
We finally note that this also works in the border setting if we let $L_1, \ldots, L_m$ be border polynomials approximating $f$.
\end{proof}

\begin{proof}[Proof of \Cref{clm:char_two_deg_two}]
Without loss of generality, we assume that $f = \sum_{i=1}^nx_iy_i$ (we can make any polynomial through a change of variables).
Let $\omega \in \FF$ be the $3$rd order primitive root of unity.
Then, observe
\[
\esym{3}{2}(\omega x + \omega^2 y, \omega^2x + \omega y, x + y) = xy,\ \esym{3}{1}(\omega x + \omega^2 y, \omega^2x + \omega y, x + y) =0.
\]
Hence, $\esym{3n}{2}(\omega x_1 + \omega^2 y_1, \omega^2x_1 + \omega y_1, x_1 + y_1, \ldots , \omega x_n + \omega^2 y_n, \omega^2x_n + \omega y_n, x_n + y_n) = \sum_{i=1}^nx_iy_i$ and $\esym{3n}{1}(\omega x_1 + \omega^2 y_1, \omega^2x_1 + \omega y_1, x_1 + y_1, \ldots , \omega x_n + \omega^2 y_n, \omega^2x_n + \omega y_n, x_n + y_n) = 0$.
\end{proof}

\section{Omitted Proofs for Symmetric Model}\label{sec:sym_model_omit}
In this section, we will prove \Cref{thm:esm_rep} in its full generality: for arbitrary fan-in and field characteristic.

To begin with, we will extend Claim \ref{clm:char_two_esp_to_reduc} to the case of arbitrary positive field characteristic.
\begin{claim}\label{clm:esr_to_reduce}
Let $p \coloneq \operatorname{char}(\FF) \ne 0$.
If $f \in \FF[x_1, \ldots, x_n]$ is homogeneous degree-$(p+1)$ and in Sym, then there are reducible homogeneous degree-$(p+1)$ polynomials $g_1, \ldots, g_{p-1} \in \FF[x_1, \ldots, x_n]$ and linear $q_1, \ldots, q_\ell \in \FF[x_1, \ldots, x_n]$ such that
\[
f = g_1 + \dots + g_{p-1} + q_1^{p+1} + \dots + q_\ell^{p+1},
\]
for some $\ell \in \NN$.
\end{claim}
\begin{proof}
Let $L_1, \ldots, L_m \in \FF[x_1, \ldots, x_n]$ be such that $\esym{m}{p+1}(L_1, \ldots, L_m) = f$.
Then, by the Newton identities (see (\ref{eq:newton_identity})), we have that
\begin{align*}
f &= \sum_{i=1}^{p+1}(-1)^{i-1}\esym{m}{p+1-i}(L_1, \ldots, L_m)p_i^m(L_1, \ldots, L_m) \\
&= \sum_{i=2}^{p}(-1)^{i-1}\esym{m}{p+1-i}(L_1, \ldots, L_m)p_i^m(L_1, \ldots, L_m) + (\esym{m}{1}(L_1, \ldots, L_m))^{p+1} + L_1^{p+1} + \dots + L_m^{p+1}
\end{align*}
\end{proof}
Notice that this claim is not as strong as the claim for characteristic two.
Currently, it is not known how to extend it, but only one direction is necessary for the rest of the proof.

We then get the following as a corollary.
\begin{corollary}
Let $p \coloneq \operatorname{char}(\FF) \ne 0$.
If a homogeneous degree-$(p+1)$ polynomial $f \in \FF[x_1, \ldots, x_n]$ is in $\Sigma^{[k]}\operatorname{Sym}$, then $s_{\operatorname{mult}}(f) \le (p-1) \cdot k$.
Further, if $f$ is in $\overline{\Sigma^{[k]}\operatorname{Sym}}$, then $\overline{s_{\operatorname{mult}}}(f) \le (p-1) \cdot k$.
\end{corollary}

We are now ready to complete the proof for the positive characteristic case.
\begin{claim}
Let $p \coloneq \operatorname{char}(\FF) \ne 0$. The polynomial $f \in \FF[x_1, \ldots, x_n]$ defined by
\[
f = \sum_{i=1}^\ell\prod_{j=(p+1)(i-1)+1}^{(p+1)i}x_j
\]
is not in $\overline{\Sigma^{[k]}\text{Sym}}$ for $\ell > k \cdot (p-1)$.
\end{claim}
\begin{proof}
Assume that $\ell > k \cdot (p-1)$ and let $g_1, \ldots, g_m \in \FF[x_1, \ldots, x_n]$ be reducible homogeneous degree-$(p+1)$ polynomials such that
\[
f \simeq \operatorname{MULT}(g_1 + \dots + g_k).
\]
We will continue with the proof in the non-border setting, and we will observe in the end that it also works in the border setting.

We will consider splitting each $g_i$ into the product of two polynomials, say $g_i = h_i \cdot p_i$, where $\deg(h_i) \le \deg(p_i)$.
For each $i \in [k]$, let $d_i$ be the degree of $h_i$.
Hence, letting $S_d \subseteq \NN$ be the set of $n$-tuples summing to $d$, i.e., $S_d = \{a \in \NN^n \mid a_1 + \dots a_n = d\}$, we can write $g_i$ as
\[
g_i = \left(\sum_{\alpha \in S_{d_i}}a_\alpha^{(i)}x^\alpha\right)\left(\sum_{\alpha \in S_{p+1-d_i}}b_\alpha^{(i)}x^\alpha\right).
\]
Consider a fixed monomial $x_{i_1} \dots x_{i_{p+1}}$, and let $M = \{i_1, \ldots, i_{p+1}\}$.
We will abuse notation and sometimes write, given a subset $S \subseteq [n]$, $a_S^{(i)}$ or $b_S^{(i)}$ to represent the coefficient corresponding to the multilinear monomial given by the set $S$.
Thus, we have that the coefficient of this monomial in $f$ is given by
\[
\sum_{i=1}^k\sum_{S \in {M \choose d_i}}a_{S}^{(i)}b_{M \setminus S}^{(i)},
\]
which we will denote by $c_M$.

We will now show that
\[
F \coloneq \sum_{\mathcal{I} \in \mathcal{P}_{p+1}[n]}\prod_{I \in \mathcal{I}}c_I \equiv 0 \mod p,
\]
where $\mathcal{P}_{p+1}[n]$ denotes the set of all ways to partition $[n]$ into groups of size $p+1$.
Once we show this, the proof is complete.

We will fix a monomial in $F$.
Observe that it can be written as $a_{S_1}^{(j_1)}b_{S_1'}^{(j_1)} \dots a_{S_\ell}^{(j_\ell)}b_{S_\ell'}^{(j_\ell)}$, where $S_1 \cup S_1' \cup \dots \cup S_\ell \cup S_\ell' = [n]$ and $t_1, \ldots, t_\ell \in [k]$.
We will consider partitions $\mathcal{I} \in \P_{p+1}[n]$ such that $\prod_{I \in \mathcal{I}}c_I$ contains $F$.
We notice that these are exactly characterized by corresponding bijections $\pi : [\ell] \rightarrow [\ell]$ such that each $S \in \mathcal{I}$ can be written as $S_i \cup S_{\pi(i)}'$ for some $i \in \ell$ and $t_j = t_{\pi(j)}$ for all $j \in [\ell]$.
Further notice that the coefficient of $F$ in these polynomials is one.
For $i \in [k]$, let $n_i$ be the number of $t_j$ values such that $t_j = i$, observing that $n_1 + \dots + n_k = \ell$.
Then, the number of such selections of bijections is exactly equal to $n_1! \cdots n_k!$.
Because $\ell > k \cdot (p-1)$, there is a $i \in [k]$ such that $n_i \ge p$.
Thus, we have shown that the coefficient of this monomial is zero modulus $p$.
We finally observe that we have created a polynomial identity in the polynomial coefficients, meaning the proof works in the border setting.
\end{proof}

\section{Reduction from $\Sigma\Pi\Sigma$ to Symmetric Model}\label{sec:reduction}
In this section, we will connect the border symmetric computational model to $\overline{\Sigma^{[k]}\Pi\Sigma}$ formulas.
In \cite{Kum20}, the main result came from observing we could convert a slightly restrictive extension of the border symmetric model to $\overline{\Sigma^{[2]}\Pi\Sigma}$ formulas.
Specifically, for a given homogeneous degree-$d$ $f \in \FF[x_1, \ldots, x_n]$ such that there are linear $L_1, \ldots, L_m \in \FF(\varepsilon)[x_1, \ldots, x_n]$, where $f \simeq \esym{m}{d}(L_1, \ldots, L_m)$ and $\esym{m}{k}(L_1, \ldots, L_m) = 0$ for every $k < d$, then we can write $f$ as a $\overline{\Sigma^{[2]}\Pi\Sigma}$ formula using (\ref{eq:prod_to_esp}):
\[
\prod_{i=1}^n(1 + \varepsilon \cdot L_i) - 1 = \varepsilon^d \cdot \esym{n}{d}(L_1, \ldots, L_m) + \varepsilon^{d+1}\sum_{i=d+1}^{n}\varepsilon^{i-d-1}\esym{n}{i}(L_1, \ldots, L_m).
\]

In this section, we will analyze the opposite relationship; namely, if we know a polynomial can be computed by a $\overline{\Sigma^{[k]}\Pi\Sigma}$ circuit, can we say anything about its computability in the symmetric model?
It turns out that we can show that the $\overline{\Sigma^{[k]}\Pi\Sigma}$ model is weaker than sums of the border symmetric model.
\begin{theorem}\label{thm:depth_three_to_sym}
If $f \in \FF[x_1, \ldots, x_n]$ is homogeneous degree-$d$ and can be represented using  $\overline{\Sigma^{[k]}\Pi\Sigma}$, then we can represent $f$ by $\overline{\Sigma^{[k]}\text{Sym}}$.
\end{theorem}
\begin{proof}
Let $f_1^{(1)}, \ldots, f_{m_k}^{(k)} \in \FF(\varepsilon)[x_1, \ldots, x_n]$ be affine and $c_1, \ldots, c_k \in \FF(\varepsilon)$ be such that $f \simeq \sum_{i=1}^kc_i\prod_{j=1}^{m_i}f_j^{(i)}$.
We claim that, if each $f_j^{(i)}$ is not constant-free, then the result is obvious.
Notice that, in this case, we can assume that each $f_j^{(i)}(0) = 1$ by dividing out each affine function's constant term.
Then, observe that
\begin{align*}
f &= H_d[f] \simeq H_d\left[\sum_{i=1}^kc_i\prod_{j=1}^{m_i}f_j^{(i)}\right] = \sum_{i=1}^k c_i \cdot H_d\left[\prod_{j=1}^{m_i}f_j^{(i)}\right] = \sum_{i = 1}^k c_i \cdot \esym{m_i}{d}(f_{1}^{(i)}, \ldots, f_{m_i}^{(i)}).
\end{align*}

Now, we only need to consider what happens if there is some $f_j^{(i)}$ such that $f_j^{(i)}(0) = 0$.
We will deal with this case using the following lemma.

\begin{lemma}\label[lemma]{lem:remove_constant}
Let $F, G, \ell \in \FF(\varepsilon)[x_1, \ldots, x_n]$ be such that $\ell$ is linear (and hence constant-free). Then, there is an $\alpha \in \FF(\varepsilon)$ such that, if $f \simeq F \cdot \ell + G$ for $f \in \FF[x_1, \ldots, x_n]$, then $f \simeq F \cdot (\ell + \alpha) + G$.
\end{lemma}
Thus, the proof is completed by iteratively applying \Cref{lem:remove_constant} to remove all of the constant-free linear forms.
Then, we simply apply the argument for the case that every $f_j^{(i)}$ is not constant-free.
\end{proof}

\begin{proof}[Proof of \Cref{lem:remove_constant}]
To complete this proof, the idea is to set $\alpha$ to a value in $\FF(\varepsilon)$ such that $F \cdot \alpha + G \in \FF[\varepsilon][x_1, \ldots, x_n]$ and is such that its lowest-degree in $\varepsilon$ is greater than the lowest-degree in $\varepsilon$ of $F \cdot \ell + G$.
Hence, including this $\alpha$ will still make this polynomial have a valid approximation and this approximation will still be correct.

To prove this, imagine that $\alpha$ is a new independent variable.
Notice that, in $F \cdot (\ell + \alpha) + G$, the coefficients of monomials in $x_1, \ldots, x_n$ are in $\FF(\varepsilon)[\alpha]$.
Hence, for each such monomial, there is a value in $\FF[\varepsilon]$ we can set $\alpha$ to that makes the coefficient an element of $\FF[\varepsilon]$.
Then, if we multiply all these values together, we ensure our selection makes the polynomial an element of $\FF[\varepsilon][x_1, \ldots, x_n]$.
Finally, we can multiply $\alpha$ by $\varepsilon^N$ for a very large value of $N$ to ensure that our selection does not change the approximation.
\end{proof}

\section{Formula Lower Bounds}\label{sec:formula_lower_bounds}
In this section, we will focus on proving the results of \cite{CKSV22}, specifically, we will prove Lemma \ref{lem:formula_lower_bound}.
This is done for the sake of completeness and because the statement that we use, while not difficult to see from the proof in the original paper, is slightly different.
Further, we provide a minor simplification to the original proof.

We will briefly recall the formal degree of a formula.
The formal degree of a leaf node is defined by the degree of the polynomial that the leaf node is labeled with.
The formal degree of a sum gate is defined as the maximum of the formal degrees of its children, and the formal degree of a product gate is defined by the sum of the formal degrees of its children.
Observe that the formal degree of a formula upper bounds the degree of the polynomial it computes.

Instead of the typical definition, we will define the size of a formula to be the number of leaves whose label is not constant.
Observe that this is linearly-related to the number of leaves in a formula and thus its size (noting that we can ``collapse" a node where both children are constants).
We do this to help simplify the proofs.

The idea behind the following proof will be to iteratively find large sub-trees in our formula whose formal degree is strictly less than the degree of the polynomial we compute.
We will then ``peel" these sub-trees from our formula and replace them with the constant part of the polynomial our sub-tree computes.
We will observe that each step in this process only adds an error term that can be represented as the product of two constant-free polynomials.
Finally, we will combine our sub-trees of low formal degree to represent our polynomial by a sum of products of constant-free polynomials summing to a low-degree polynomial.
Our conclusion will then follow using Lemma \ref{lem:kumar_lemma}.

To begin, we will state a simple fact about algebraic circuits, which merely states that every formula has a node with formal degree in the range $[t, 2t-1]$.

\begin{proposition}[See Lemma 5.11 of \cite{CKSV22}]\label{prop:formula_vertex_pick}
Let $\Phi$ be a formula of formal degree $d$. Then for each $t \in [1, d/2]$, there is a vertex $v$ in $\Phi$ such that $\Phi_v$ (the formula defined by the sub-tree rooted at $v$) has formal degree at least $t$ and at most $2t - 1$.

Further, there exists polynomials $h, f \in \FF[x_1, \ldots, x_n]$ such that $\Phi = h \cdot \Phi_v + f$ and, for every $\gamma \in \FF$, $h\gamma + f$ can be computed by a formula of size at most $|\Phi| - |\Phi_v|$.
\end{proposition}
\begin{proof}
First, notice that the formal degree of a formula monotonically increases as one goes from leaf to root. The leaves of the formula have formal degree at most one and the root has formal degree $d$.
Then, we observe that the formal degree of an internal node is at most the sum of the formal degrees of its children.
It is, therefore, impossible for a parent node to have formal degree above $2t - 1$ with children of formal degree strictly less than $t$.

Now, let $v$ be a vertex in $\Phi$ that satisfies the previous hypotheses. Notice that $\Phi$ is linear in $\Phi_v$, so there are $h, f \in \FF[x_1, \ldots, x_n]$ such that $\Phi = h\Phi_v + f$. Notice that, for any $\gamma \in \FF$, we could replace $\Phi_v$ with a leaf node labeled with $\gamma$, and it would have size $|\Phi| - |\Phi_v|$.
\end{proof}

In the next proposition, we represent our process of converting our polynomial to the form we would like.
Specifically, we suppose that we have some $d'$, and we want to represent a polynomial by a formula of formal degree strictly less than $d'$ with some error, represented by the sum of products of constant-free polynomials.
We will repeatedly apply Proposition \ref{prop:formula_vertex_pick} to complete this simplification.

\begin{proposition}[See Lemma 5.12 of \cite{CKSV22}]\label{prop:formal_dgree_reduce}
Let $\Phi$ be a formula of size $s$ and formal degree $d$. For every $d' \ge 3$, there is a formula $\Phi'$ and polynomials $f_1, \ldots, f_k, g_1, \ldots, g_k \in \FF[x_1, \ldots, x_n]$ such that
\begin{itemize}
    \item $\Phi = \Phi' + \sum_{i=1}^kf_ig_i$,
    \item $\Phi'$ has formal degree less than $d'$,
    \item $f_1, \ldots, f_k, g_1, \ldots, g_k$ are constant-free polynomials, and
    \item $k \cdot \frac{d'}{3} \le s$.
\end{itemize}
\end{proposition}
\begin{proof}
We will prove this inductively on the size of the formula $s$.
We first observe that the claim is trivial if $d < d'$ (which certainly happens when the circuit is sufficiently small), as we can set $\Phi' = \Phi$.
Thus, if $s$ is small enough such that $d < d'$, the claim is obvious.

Now, assume the hypothesis is true for all circuit of size strictly smaller than $s$.
We assume that $d \ge d' \ge 3$. Then, we apply Proposition \ref{prop:formula_vertex_pick} so that $v$ is a vertex in $\Phi$ such that $\Phi_v$ has formal degree between $d'/3$ and $2d'/3$ and let $h, f \in \FF[x_1, \ldots, x_n]$ satisfy the rest of the proposition. We will write $h = h' + \alpha$ and $\Phi_v = g' + \beta$, where $h', g' \in \FF[x_1, \ldots, x_n]$ are constant-free polynomials and $\alpha, \beta \in \FF$. Then,
\[
\Phi = h\Phi_v + f = (h' + \alpha)(g' + \beta) + f = h'g' + \alpha g' + (h\beta + f).
\]
Notice that $\alpha g'$ can be computed by a formula of size at most $|\Phi_v|$ and formal degree at most $2d'/3$ and, by Proposition \ref{prop:formula_vertex_pick}, $h\beta + f$ can be computed by a formula of size at most $|\Phi| - |\Phi_v|$.
We can thus apply the inductive hypothesis to $h\alpha + f$ and let $\Phi'$ be a formula and $f_1, \ldots, f_k, g_1, \ldots, g_k \in \FF[x_1, \ldots, x_n]$ be polynomials that satisfy the hypotheses. We have that $h', g'$ are constant-free, so we let $f_{k+1} = h'$ and $g_{k+1} = g'$. We define $\Phi''$ to be the circuit resulting from summing $\alpha g'$ and $\Phi'$, so that $\Phi''$ has size at most $s$ and formula degree at most $d'$. Then,
\[
\Phi = \Phi'' + \sum_{i=1}^{k+1}f_ig_i.
\]
Further notice that (observing that $|\Phi_v| \ge d'/3$)
\[
(k+1) \cdot \frac{d'}{3} \le |\Phi| - |\Phi_v| + \frac{d'}{3} \le |\Phi|.\qedhere
\]
\end{proof}

We can now use this to conclude the main result of this section.

\begin{proof}[Proof of Lemma \ref{lem:formula_lower_bound}]
Suppose $\Phi$ has formal degree $D$. Then, we apply Proposition \ref{prop:formal_dgree_reduce} to obtain polynomials $h, f_1, \ldots, f_k, g_1, \ldots, g_k \in \FF[X]$ such that
\begin{itemize}
    \item $f = p + \sum_{i=1}^kf_ig_i$,
    \item $\deg(p) < d$,
    \item $f_1, \ldots, f_k, g_1, \ldots, g_k$ are constant-free, and
    \item $k\cdot \frac{d}{3} \le s$.
\end{itemize}

Then, if we consider $\Phi$ as a formula over the algebraic closure of $\FF$, we can apply Proposition \ref{lem:kumar_lemma} to conclude that
\[
\dim V_2(f) \ge n - 2k \ge n - 2s \cdot \frac{3}{d}.
\]
Therefore,
\[
s \ge \frac{d}{6}(n - \dim V_2(f)).\qedhere
\]
\end{proof}

\end{document}